\documentclass[12pt,oneside,reqno]{amsart}

\usepackage{amsmath}
\usepackage{ifthen}
\usepackage{amssymb}
\usepackage{latexsym}
\usepackage{amsfonts,setspace}
\usepackage{fullpage}
\usepackage{indentfirst}
\usepackage{algorithm}
\usepackage{color}
\usepackage{changepage}
\usepackage{paralist}
\usepackage{tikz}
\usepackage{bm}
\usepackage{hyperref}

\usepackage{makecell}
\usepackage{pict2e}
\usepackage{algorithmic}
\usepackage{amsthm}
\theoremstyle{definition}
\newtheorem{thm}{Theorem}
\newtheorem{rmk}[thm]{Remark}
\newtheorem{prop}[thm]{Proposition}
\newtheorem{cor}[thm]{Corollary}
\newtheorem{lem}[thm]{Lemma}
\newtheorem{exm}[thm]{Example}
\newtheorem{defi}[thm]{Definition}

\newcommand{\anton}[1]{{\color{red} (Anton: #1)}}

\newcommand{\RR}{\mathbb{R}}
\newcommand{\CC}{\mathbb{C}}

\newcommand{\KK}{\mathbb{K}}

\newcommand{\appx}{{\widetilde x}}
\newcommand{\pointnorm}[1]{\|(1,#1)\|}

\DeclareMathOperator{\erf}{erf}

\newcommand{\sage}{\texttt{SageMath}}

\newcommand{\mainfilecheck}[1]{0}

\normalbaselines

\def\Dfinite{$D$-finite}
\title{Effective certification of approximate solutions to systems of
  equations involving analytic functions}
\author{Michael Burr, Kisun Lee, Anton Leykin}

\begin{document}

\begin{abstract}
We develop algorithms for certifying an approximation to a nonsingular solution of a square system of equations built from univariate analytic functions.  These algorithms are based on the existence of oracles for evaluating basic data about the input analytic functions.  One approach for certification is based on $\alpha$-theory while the other is based on the Krawczyk generalization of Newton's iteration. We show that the necessary oracles exist for \Dfinite\ functions and compare the two algorithmic approaches for this case using our software implementation in \sage.
\end{abstract}

\maketitle

\section{Introduction}\label{sec:introduction}

The main problem that we consider in this paper is to {\em certify} an approximation of a nonsingular root, i.e., a root of the multiplicity $1$ of a function $F:U \rightarrow\CC^n$ (or $\RR^n$), where $U\subset\CC^n$ (or $\RR^n$) is an open subset, corresponding to the following square system of equations:
\begin{equation*}
F(x)=\begin{bmatrix}
F_1(x_1,\dots, x_n)\\
\vdots\\
F_n(x_1,\dots, x_n)
\end{bmatrix} = 0.
\end{equation*}
The minimal goal of certification is, given (a finite description of) a compact region $I\subset\CC^n$ (or $\RR^n$) that is conjectured to contain a unique root, execute an algorithm which produces a certificate for the existence and uniqueness of a root in $I$.  The algorithms considered in this paper are based on {\em $\alpha$-theory} and the {\em Krawczyk test}, which, in turn, are based on Newton iteration. 

Throughout this paper, we focus on the theory for the case of $\CC$ and provide a few remarks identifying the changes that must be made to use these techniques in $\RR$.  

Recall that the Newton operator $N_F(x)$ for a square differentiable function $F$ is defined as 
$$
N_F(x)=\begin{cases}x-F'(x)^{-1}F(x)&\text{if $F'(x)$ is invertible and}\\x&\text{otherwise}\end{cases},
$$
where $F'(x)$ is the Jacobian of $F$ at $x$.  Moreover, the fixed points of $N_F(x)$ correspond to roots of $F$ or where $F'(x)$ fails to be invertible.  Additionally, for $\tilde{x}$ sufficiently close to a nonsingular root $x^\ast$, the {\it $k$-th Newton iteration} $N_F^k(\tilde{x})$, defined by applying the operator $k$ times, converges to $x^\ast$.  For more details, see, e.g., \cite[Chapter 8]{blum2012complexity}.  Since we focus on the Newton operator, we do not discuss alternate approaches for certification which use global methods or do not use fixed points, see, e.g., \cite{MooreKioustelidis}.

The approach of {\em $\alpha$-theory} starts with a suspected approximate solution $\appx$ and attempts to construct a ball containing both $\appx$ and a unique root of $F$ with the guarantee that, starting at any point in the ball, {\em Newton's method} converges quadratically to a root of $F$.  The $\alpha$-theory-based approach uses point estimates on the value of $F$ and (all of) its derivatives.  We develop explicit estimates of this form in this paper as practical and algorithmic extensions of an {\em $\alpha$-test}, presented in its initial setting by Smale~\cite{smale1986newton}.

The {\em Krawczyk operator} is an interval-based certified generalization of the Newton operator.  The Krawczyk operator starts with an $n$-dimensional interval $I$ containing a suspected root and uses interval arithmetic to determine whether a Newton-like operator is contractive within $I$.  In this case, all points of $I$ converge to a unique fixed point without a guarantee on their convergence speed.  The Krawczyk operator approach uses estimates on the value of $F'$ over $I$.  We adapt a version of the Krawczyk operator that provides a {\em Krawczyk test}, developed in its initial setting by Krawczyk~\cite{krawczyk1969newton}.

The main contributions of this paper are descriptions of procedures to certify nonsingular roots of square systems.  These procedures are described in terms of oracles, and in cases where these oracles exist, our theoretical tests based both on $\alpha$-theory and the Krawczyk operator lead to certification algorithms.  We show that these oracles can be implemented for systems of equations involving \Dfinite\ functions. In particular, we extend the range of effective application of $\alpha$-theory.  We implement our algorithms in \sage\ and provide experimental analyses of our two approaches. Implementations and computation examples are available at\\
\nopagebreak[4]
\centerline{\url{https://github.com/klee669/DfiniteComputationResults}}

\subsection{Setting}

The theory originally derived for $\alpha$-theory \cite{smale1986newton} and the Krawczyk operator \cite{krawczyk1969newton} applies to arbitrary square systems of analytic functions.  Both approaches, however, require various computations and data which are typically inaccessible for arbitrary analytic functions, e.g., the $\gamma$-function in $\alpha$-theory can be based on an infinite number of derivatives.  Therefore, considerable work has gone into finding cases where these tests can be applied algorithmically.

We describe the classes of functions by seeding a class with a set of {\em basic functions} (we informally call them {\em ingredients}) and then extending it by recursively applying the basic arithmetic operations (addition and multiplication) to the basic functions and constants finitely many times.   By adding more variables and equations, other operations, such as division and composition, are possible in the construction.  

More explicitly, suppose that the basic functions include both the coordinate functions and additional basic functions $\{g_1,\dots,g_m\}$.  Then, the systems of equations that we construct can be written in the following form (after an appropriate change of variables):
\begin{equation}\label{eq:SystemWithIngredients}
F(x):=\begin{bmatrix}
p_1(x_1,\dots,x_{n+m})\\
\vdots\\
p_n(x_1,\dots,x_{n+m})\\
x_{n+1}-g_1(x_1)\\
\vdots\\
x_{n+m}-g_m(x_m)
\end{bmatrix}
\end{equation}
where $p_i\in \CC[x_1,\dots, x_{n+m}]$ for $i=1,\dots, n$.

Suppose that the basic functions are the coordinate functions $\{x_1,\dots,x_n\}$, then the class of functions is $\CC[x]=\CC[x_1,\dots,x_n]$, i.e., the class of polynomial systems of equations.  This class appears frequently in geometric problems (e.g.,\cite{bozoki2015seven}) and can be effectively studied via $\alpha$-theory and the Krawczyk operator since all but finitely many derivatives vanish. For a practical implementation of the $\alpha$-theory approach in this setting, see~\cite{hauenstein2012algorithm}.
	
When the basic functions are the coordinate functions along with univariate analytic functions which satisfy linear differential equations with constant coefficients, the resulting class of functions are the polynomial-exponential functions.  In \cite{hauenstein2017certifying}, Hauenstein and Levandovskyy extend $\alpha$-theory-based certification to this case.

In this paper, we take our general approach and apply it to the class of functions built from the coordinate functions and \Dfinite\ functions.  We recall that a \textit{$D$-finite function} $g$ is a solution to a linear differential equation with polynomial coefficients $p_k(t)\in \CC[t]$, i.e., a differential equation of the following form:
\begin{equation}\label{eq:Dfinite}
p_r(t)g^{(r)}(t)+\cdots p_1(t)g'(t)+p_0(t)g(t)=0.
\end{equation}
If $p_r(0)$ does not vanish, then there is a unique function $g(t)$ which satisfies both Equation (\ref{eq:Dfinite}) and specified initial conditions $g(0)=c_0$, $g'(0)=c_1$, $\dots$, and $g^{(r-1)}(0)=c_{r-1}$. We call the corresponding class of functions {\em polynomial-\Dfinite\ functions}.  The generalization of effective $\alpha$-theory-based algorithms to this larger class of functions is one of the main advances of this paper.

\subsection{Paper organization}
The structure of the remainder of this paper is as follows: We recall and present the general theory for certifying solutions to systems of analytic functions using the Krawczyk operator and $\alpha$-theory in \S\S\ref{sec:Krawczyk} and \ref{sec:alpha-theory}, respectively.  In particular, we explicitly describe the oracles which are needed for the application of these tests.  In \S\ref{sec:Dfinite}, as an example, we illustrate how these oracles exist for \Dfinite\ functions.  Software implementation, examples, and applications are discussed in \S\ref{sec:examples}.  The framework that we introduce in this paper has the potential to be applicable to other families of basic functions.  We present remarks on the potential development of these techniques in~ \S\ref{sec:conclusion}.




\section{Certification using the Krawczyk Method}\label{sec:Krawczyk}


In this section, we develop the theory of interval arithmetic and the Krawczyk operator, an interval-based generalization of the Newton operator.  We explicitly describe the oracles which are necessary so that the theory described in this section can be developed into an algorithm.  In \S\ref{sec:Dfinite}, we show that these oracles exist for \Dfinite\ functions, and, so, the Krawczyk operator can be used to certify roots of \Dfinite\ functions.

\subsection{Interval Arithmetic}

Interval arithmetic performs conservative computations with intervals in order to produce certified computations.  For example, suppose that $[a,b]$ and $[c,d]$ are isolating intervals for $x,y\in\mathbb{R}$, i.e., $x\in[a,b]$ and $y\in[c,d]$.  Then, interval arithmetic formalizes the conclusion that $x+y\in[a+c,b+d]$.  More precisely, for any arithmetic operation $\odot$,
$$
[a,b]\odot[c,d]=\{x\odot y:x\in[a,b],y\in[c,d]\}.
$$
For the standard arithmetic operations, there are formulas for the interval versions of these operators, see, e.g., \cite{moore2009introduction} for more details.  

These methods can extend to complex numbers by writing intervals in $\CC$ as $[a_1,a_2]+[b_1,b_2]i$.  In this case, multiplication of complex interval numbers is computed as 
\ifthenelse{\mainfilecheck{1} > 0}{\begin{multline}
	([a_1,a_2]+[b_1,b_2]i)([c_1,c_2]+[d_1,d_2]i)\\=([a_1,a_2][c_1,c_2]-[b_1,b_2][d_1,d_2]) \label{eq:complex:mult}\\+([a_1,a_2][d_1,d_2]+[b_1,b_2][c_1,c_2])i.
	\end{multline}
	}{\begin{multline}
	([a_1,a_2]+[b_1,b_2]i)([c_1,c_2]+[d_1,d_2]i)\\=([a_1,a_2][c_1,c_2]-[b_1,b_2][d_1,d_2])+([a_1,a_2][d_1,d_2]+[b_1,b_2][c_1,c_2])i.\label{eq:complex:mult}
	\end{multline}
	}
We observe that the image of this product may be strictly larger than the set of possible products of elements from the pair of complex intervals.  This formulation, however, is critically important in our development of the Krawczyk method in Section \ref{sec:KrawczykMethoddetails}.

We write $\mathbb{IC}$ for the set of intervals in $\mathbb{C}$, and we write $\mathbb{IC}^n$ for the set of $n$-dimensional boxes in $\mathbb{C}^n$, i.e., $n$-fold products of intervals in $\mathbb{C}$.  For an open set $U\subseteq\CC^n$, we write $\mathbb{I}U$ for intervals in $\mathbb{IC}^n$ which are contained in $U$.  For a function $F:U\rightarrow\CC$, an oracle interval extension of $F$ is an oracle $\square F:\mathbb{I}U\rightarrow\mathbb{IC}$ such that for any $I\in\mathbb{I}U$,  
$$
\square F(I)\supseteq F(I):=\{F(x):x\in I\}.
$$
In other words, $\square F(I)$ is an interval containing the image of $F$ on $I$.  For polynomial systems, such oracles can be constructed using interval arithmetic, see, e.g., \cite{moore2009introduction,computermethods:range} for details.  We discuss the existence of such oracles for \Dfinite\ functions in \S\ref{sec:Dfinite}.

\subsection{The Krawczyk Method}\label{sec:KrawczykMethoddetails}
The Krawczyk operator combines both interval arithmetic and a generalization of the Newton operator in order to develop a certified test for an isolated root of a square system of equations in a region.  The Krawczyk operator is one member of the family of interval-based Newton-type methods, see, e.g., \cite[Chapter 8]{moore2009introduction} and the references included therein for more details.  In most presentations of the Krawczyk operator, see, e.g., \cite{krawczyk1969newton,moore2009introduction}, the operator is only described for real variables.  There are some subtle differences that arise in the complex setting; therefore, in this section, we provide the theory for the Krawczyk operator for complex variables.

Suppose that, for an open set $U\subset\CC^n$, $F:U\rightarrow\mathbb{C}^n$ is a square differentiable system of functions and let $Y\in GL_n$ the set of $n\times n$ invertible matrices.  We observe that $F(x)=0$ if and only if $x$ is a fixed point of $G(x):=x-YF(x)$.  We note that if $Y$ were replaced by $F'(x)^{-1}$, then this function would be the Newton operator.  The correspondence between the fixed points of $G$ and the roots of $F$ is the motivation for the Krawczyk operator:
\begin{defi}
Let $U\subset\CC^n$ be an open set and $F:U\rightarrow\mathbb{C}^n$ be a square differentiable system of functions such that $F'$ has an interval extension $\square F'$.  Let $y\in I\in\mathbb{I}U$ and $Y\in GL_n$.  The {\em Krawczyk operator} centered at $y$ is defined to be
$$
K_y(I):=y-YF(y)+(I_n-Y\square F'(I))(I-y),
$$
where $I_n$ is the $n$-dimensional identity matrix.
\end{defi}
When the domain and codomain are real, the Krawczyk operator is an interval extension of the function $G$ using the {\em mean value form}, see, e.g., \cite[Chapter 6]{moore2009introduction}.  In the complex case, however, there is no mean value theorem, but with the definition of complex multiplication for intervals from Equation (\ref{eq:complex:mult}), the Krawczyk operator remains an interval extension of the function $G$.
\begin{lem}\label{lem:inclusion}
Let $U\subset\CC^n$ be an open set and $F:U\rightarrow\mathbb{C}^n$ be a square differentiable system of functions such that $F'$ has an interval extension $\square F'$.  Let $y\in I\in\mathbb{I}U$ and $Y\in GL_n$.  Then,
$$
G(I)\subseteq K_y(I).
$$
\end{lem} 
\begin{proof}
We observe that $K_y(I)=G(y)+(I_n-Y\square F'(I))(I-y)$, so it is enough to show that for any $z\in I$, $G(z)-G(y)\in(I_n-Y\square F'(I))(I-y)$.  Let $w=\Re(z)+i\Im(y)$; we note that $w\in I$ since $I$ is a rectangle.  Then, we consider the real path form $y$ to $w$ and the purely imaginary path from $w$ to $z$.  Considering these two paths as functions of a real variable, we use the mean value theorem on each path and on the real and imaginary parts of $G$ separately.  Fix $1\leq j\leq n$.  After applying the Cauchy-Riemann equations, each of $G_j(w)-G_j(y)$ and $G_j(z)-G_j(w)$ can be written in terms of the real and imaginary parts of $G_j'$ at appropriate points times $(w-y)$ or $(z-w)$.  Then, the sum of these two formulae correspond to elements of the four products appearing in Equation (\ref{eq:complex:mult}).  By repeating this computation for each $j$, we conclude that $G(z)\in K_y(I)$.
\ifthenelse{\mainfilecheck{1}>0}{}{We begin by observing that $I_n-Y\square F'(I)$ is an interval matrix containing $G'(I)$.  Our plan, for a fixed $z\in I$, is to write $G(z)-G(y)$ in terms of elements of $G'(I)$, $\Re(z-y)$, and $\Im(z-y)$ in order to conclude the desired containment.

Let $w=\Re(z)+i\Im(y)$, and consider the path from $y$ to $w$, which is a real path, followed by the path from $w$ to $z$, which is purely imaginary path.  Fix $1\leq j\leq n$.  By the real mean value theorem, there are some $c_1$ and $c_2$ along the line between $y$ and $w$ so that $\nabla_{\Re}(\Re G_j(c_1))\cdot(w-y)=\Re G_j(w)-\Re G_j(y)$ and $\nabla_{\Re}(\Im G_j(c_2))\cdot(w-y)=\Im G_j(w)-\Im G_j(y)$.  Here, the subscript indicates that the derivative is only being taken with respect to the real variable.  Similarly, along the line between $w$ and $z$, there are some $c_3$ and $c_4$ so that $\nabla_{\Im}(\Re G_j(c_3))\cdot \Im(z-w)=\Re G_j(z)-\Re G_j(w)$ and $\nabla_{\Im}(\Im G_j(c_4))\cdot \Im(z-w)=\Im G_j(z)-\Im G_j(w)$, where the derivative is being taken with respect to the complex variable.  Putting these together (and multiplying by $i$ as appropriate), we get
\begin{align*}
G_j(w)&=G_j(y)+\nabla_{\Re}(\Re G_j(c_1))\cdot(w-y)+i\nabla_{\Re}(\Im G_j(c_2))\cdot(w-y)\\
G_j(z)&=G_j(w)-i\nabla_{\Im}(\Re G_j(c_3))\cdot (z-w)+\nabla_{\Im}(\Im G_j(c_4))\cdot (z-w).
\end{align*}
Using the Cauchy-Riemann equations, we find that
\begin{align*}
G_j(w)&=G_j(y)+\Re G_j'(c_1)\cdot(w-y)+i\Im G_j'(c_2)\cdot(w-y)\\
G_j(z)&=G_j(w)+i\Im G_j'(c_3)\cdot (z-w)+\Re G_j'(c_4)\cdot (z-w).
\end{align*}
Therefore,
\ifthenelse{\mainfilecheck{1} > 0}{\begin{multline}\label{eq:formulaz}
	G_j(z)=G_j(y)+\Re G_j'(c_1)\cdot(w-y)+i\Im G_j'(c_3)\cdot (z-w)\\+\Re G_j'(c_4)\cdot (z-w)+i\Im G_j'(c_2)\cdot(w-y).
	\end{multline}
}{\begin{equation}\label{eq:formulaz}
G_j(z)=G_j(y)+\Re G_j'(c_1)\cdot(w-y)+i\Im G_j'(c_3)\cdot (z-w)+\Re G_j'(c_4)\cdot (z-w)+i\Im G_j'(c_2)\cdot(w-y).
\end{equation}
}

Finally, we observe that since each $c_i$ is in $I$, the real and imaginary parts of $G'_j(c_i)$ are in the $j^{\text{th}}$ row of $I_n-Y\square F'(I)$.  In addition, since $w-y=\Re(z-y)$ and $z-w=i\Im(z-y)$, it follows that the differences $w-y$ and $z-w$ are also in the corresponding real and imaginary parts of $I-y$.  Finally, the four products appearing in Equation (\ref{eq:formulaz}) correspond to elements of the four products appearing in Equation (\ref{eq:complex:mult}).  By repeating this for each $1\leq j\leq n$, independently, we conclude that $G(z)\in K_y(I)$ and the desired inclusion holds.}{}
\end{proof}

In the following theorem, we collect a few facts about detecting the existence and uniqueness of roots using the Krawczyk operator. We include the proof for completeness.

\begin{thm}[cf \cite{krawczyk1969newton}]\label{thm:krawczyk}
Let $U\subset\CC^n$ be an open set and $F:U\rightarrow\mathbb{C}^n$ be a square differentiable system of functions such that $F'$ has an oracle interval extension $\square F'$.  Let $y\in I\in\mathbb{I}U$ and $Y\in GL_n$.  The following hold:
\begin{enumerate}
\item If $x\in I$ is a root of $F$, then $x\in K_y(I)$,
\item If $K_y(I)\subset I$, then there is a root of $F$ in $I$, and
\item If $I$ contains a root of $F$ and $\sqrt{2}\|I_n-Y\square F'(I)\|<1$, then the root in $I$ is unique.  Here, $\|I_n-Y\square F'(I)\|$ denotes the maximum operator norm of a matrix in $I_n-Y\square F'(I)$ under the max-norm.
\end{enumerate}
\end{thm}

\begin{proof} \begin{inparaenum}[(1)]
\item Since $x$ is a fixed point of the function $G$ if and only if $x$ is a root of $F$, by the properties of interval extensions, if $x\in I$ is a root of $F$, then $G(x)=x$ is in $K_y(I)$.
\item If $K_y(I)\subset I$, then the image of the function $G$ on $I$ is a subset of $I$, so, by Brouwer's fixed point theorem, $G$ has a fixed point, i.e., a root of $F$.
\item We observe that by expanding the proof of Lemma \ref{lem:inclusion}, we find that for all $z_1,z_2\in I$, 
$$
G(z_1)-G(z_2)\in\square G'(I)\cdot\Re(z_1-z_2)+\square G'(I)\cdot\Im(z_1-z_2).
$$
Thus,
\ifthenelse{\mainfilecheck{1} > 0}{\begin{multline}
	\|G_1(z_1)-G(z_2)\|_\infty\leq \|I_n-Y\square F'(I)\|\|\Re(z_1-z_2)\|_\infty\\+\|I_n-Y\square F'(I)\|\|\Im(z_1-z_2)\|_\infty.
	\end{multline}
}{$$
\|G_1(z_1)-G(z_2)\|_\infty\leq \|I_n-Y\square F'(I)\|\|\Re(z_1-z_2)\|_\infty+\|I_n-Y\square F'(I)\|\|\Im(z_1-z_2)\|_\infty.
$$
}
The Cauchy-Schwartz inequality and the assumption imply that 
$$\|G_1(z_1)-G(z_2)\|_\infty\leq\sqrt{2}\|I_n-Y\square F'(I)\|\|z_1-z_2\|_\infty<\|z_1-z_2\|_\infty,$$
and we conclude that the $G$ function is contractive within $I$.
\end{inparaenum}\end{proof} 

\begin{rmk}
The results of Theorem~\ref{thm:krawczyk} apply when $\CC$ is replaced by $\RR$.  In fact, in the case of $\RR$, the uniqueness test simplifies to $\|I_n-Y\square F'(I)\|<1,$ i.e., without the $\sqrt{2}$ factor.
\end{rmk}

Theorem~\ref{thm:krawczyk} serves as a proof of correctness of the following algorithm.

\begin{algorithm} {\textbf{KrawczykTest$(F,I,Y,y,\square F')$:}}
\renewcommand{\algorithmicrequire}{\textbf{Input:}}
\renewcommand{\algorithmicensure}{\textbf{Output:}}
\begin{algorithmic}
  \REQUIRE{A square differentiable system of functions $F:U\rightarrow \CC^n$ for an open set $U\subset \CC^n$, an interval $I\in \mathbb{I}U$, an invertible matrix $Y\in GL_n$, a point $y\in I$ and an interval extension $\square F'$.}
  \ENSURE{The boolean value of a condition that implies that ``the interval $I$ contains a unique nonsingular root $x$ of $F$''.}
  \smallskip
  \hrule
  \smallskip
  \RETURN $K_y(I)\subset I$ \AND $\sqrt{2}\|I_n-Y\square F'(I)\|< 1$
\end{algorithmic}
\end{algorithm}

In practice, the preconditioning matrix $Y$ is chosen to make $\|I_n-Y\square F'(I)\|$ as small as possible.  Without additional information, a good choice is often an approximation to $F'(m(I))^{-1}$, provided it exists, along with $y=m(I)$, i.e., the midpoint of $I$.  

We also observe that it might not be possible to evaluate $F(y)$ exactly.  Therefore, we consider a generalization of the Krawczyk operator.  Suppose that there is an oracle $\square F$ which, on input $y\in\CC^n$, returns an interval $\square F(y)$ containing $F(y)$.  Then, we may replace $F(y)$ by $\square F(y)$ in the definition of the Krawczyk operator as follows:
\begin{equation}\label{eq:krawczykinterval}
\square K_{y}(I)=y-Y\square F(y)+(I_n-Y\square F'(I))(I-y).
\end{equation}
We observe that $K_y(I)\subset \square K_{y}(I)$.  Therefore, when the corresponding existence and uniqueness results hold for $\square K_y(I)$, they also hold for $K_y(I)$.  By combining this operator with Theorem \ref{thm:krawczyk}, we arrive at a certified test for the Krawczyk operator.  In particular, checking that both $\square K_{y}(I)\subset I$ and $\sqrt{2}\|I_n-Y\square F'(I)\|<1$ hold, we certify that $I$ contains a unique root of $F$.  In this case, any point of $I$ approximates the root of $F$ in $I$.



\section{Certification using $\alpha$-theory}\label{sec:alpha-theory}

In this section, we introduce an effective extension of $\alpha$-theory for analytic functions.  We explicitly describe the oracles necessary for the theory to be developed into an algorithm.

\subsection{Smale's \texorpdfstring{$\bm{\alpha}$}{alpha}-theory}

In this section, we recall Smale's $\alpha$-theory, which is used to certify the solutions of square systems of analytic functions.  Let $F:U\rightarrow\CC^n$ be a square system of analytic functions defined on open set $U\subset \CC^n$. 
\emph{Quadratic convergence} of $\{N_F^k(x)\}$ to a solution of $F$ is defined as follows:
\begin{defi}\label{def:approxSolution}
  A point $x\in \CC^n$ is called \textit{an approximate solution} to $F$ with \textit{associated solution} $x^*$ with $F(x^\ast)=0$ if for every $k\in \mathbb{N}$,
  \[\left\|N_F^k(x)-x^*\right\|\leq \left(\frac{1}{2}\right)^{2^k-1}\|x-x^*\|.\]
Moreover, if $F'(x)$ is not invertible, then $x$ is an approximate solution if and only if $F(x)=0$. 
\end{defi}
$\alpha$-theory provides a certificate for a point $x$ to be an approximate solution to $F$ using three values: $\alpha(F,x),\beta(F,x)$ and $\gamma(F,x)$. If $F'(x)$ is invertible, we define
\begin{equation*}
	\begin{array}{ccl}
	\alpha(F,x)& := & \beta(F,x)\gamma(F,x)\\
	\beta(F,x) & := & \|x-N_F(x)\|=\|F'(x)^{-1}F(x)\|\\
	\gamma(F,x) & := & \sup\limits_{k\geq 2}\left\|\frac{F'(x)^{-1}F^{(k)}(x)}{k!}\right\|^{\frac{1}{k-1}}
	\end{array}
\end{equation*}
where $F^{(k)}(x)$ in the definition of $\gamma(F,x)$ is a symmetric tensor whose components are the $k$-th partial derivatives of $F$, see \cite[Chapter 5]{MR783635}. The norm in $\beta(F,x)$ is the usual Euclidean norm and the norm in $\gamma(F,x)$ is the operator norm on $S^k\mathbb{C}^n$ (for details, see \cite{hauenstein2011alphacertified}).  When $F'$ is not invertible at $x$, we define $\alpha(F,x)=\beta(F,x)=\gamma(F,x)=\infty$, but we do not consider this case in this paper.  The following theorem is the main theorem of $\alpha$-theory:
\begin{thm}(\cite[Theorem 2]{hauenstein2012algorithm}) \label{thm:alphatheory}
Let $F:\CC^n\rightarrow \CC^n$ be a system of analytic functions, and let $x$ be any point in $\CC^n$.  If 
  \[\alpha(F,x)<\frac{13-3\sqrt{17}}{4},\]
  then $x$ is an approximate solution for $F$. Moreover, $\|x-x^*\|\leq 2\beta(F,x)$ where $x^*$ is the associated solution to $x$.
\end{thm}

Moreover, with a stricter test, $\alpha$-theory also provides a way to identify when other points approximate the same root of $F$.  This is expressed in the following theorem:
\begin{thm}[{\cite[Theorem 4 and Remark 6, Chapter 8]{blum2012complexity}}]\label{thm:gammatheory}
Let $F:\CC^n\rightarrow\CC^n$ be a system of analytic functions, and let $x$ be any point in $\CC^n$.  If
$$\alpha(F,x)<0.03\quad\text{and}\quad\|x-y\|<\frac{1}{20\gamma(F,x)},$$ then $x$ and $y$ are approximate solutions to the same root of $F$. Also, there is a unique root $x^*$ of $F$ in the ball centered at $x$ with radius $\frac{1}{20\gamma(F,x)}$. Furthermore, if $\|x-\overline{x}\|>4\beta(F,x)$, then $x^*$ is not real.
\end{thm}

\begin{rmk}
The results of Theorems \ref{thm:alphatheory} and \ref{thm:gammatheory} apply when $\CC$ is replaced by $\RR$.  In particular, if $x$ is real and both $F$ and $F'$ are real-valued over $\RR$, then, when the hypotheses of these theorems are satisfied, the corresponding root of $F$ is real. 
\end{rmk}

We observe that in many cases, $\beta$ can be explicitly computed or bounded.  For example, suppose there are oracles $\square F(x)$ and $\square F'(x)$ that return intervals or boxes containing $F(x)$ and $F'(x)$.  Then, $\beta(F,x)$ can be estimated by bounding $\square F'(x)^{-1}\square F(x)$.  In Section \ref{sec:Dfinite}, we show that such oracles exist for \Dfinite\ functions.  Therefore, throughout the remainder of this section, we focus on bounding the value of $\gamma(F,x)$.

\subsection{Bounds on $\gamma$ for polynomial systems}\label{sec:alpha:polynomial}

In most applications of $\alpha$-theory the key step is to compute (or bound) $\gamma$.  In this section, we recall the construction in \cite[Section I-3]{shub2000complexity} for the case where $F=P$ is a square polynomial system, i.e., $m=0$ in Equation (\ref{eq:SystemWithIngredients}).  These bounds are needed for the polynomial part for the general case of Equation (\ref{eq:SystemWithIngredients}).

For a polynomial $p=\sum_{|\nu|\leq d}a_\nu x^\nu$, we recall that the Bombieri-Weyl norm is defined as 
\begin{equation*}
    \|p\|^2=\frac{1}{d!}\sum\limits_{|\nu|\leq d}\nu!(d-|\nu|)! |a_\nu|^2.
  \end{equation*}
For a system of polynomials $P=(p_1,\dots,p_n)$, we define 
$$
\|P\|^2=\sum_{i=1}^n\|p_i\|^2.
$$
Moreover, we let $d_i=\deg p_i$ be the degree of the $i^{\text{th}}$ polynomial and $d=\max d_i$ be the maximum degree of the polynomials.  For a point $x\in\CC$, we denote $1+\sum_{i=1}^n|x_i|^2$ by $\|(1,x)\|^2$, and we let $\Delta_P(x)$ be the diagonal matrix with entries 
$$
\Delta_P(x)_{ii}:=\sqrt{d_i} \pointnorm{x}^{d_i-1}.
$$
With these definitions in hand, we may use them to bound $\gamma$ for a polynomial system as follows:

\begin{prop}[{\cite[Proposition 5]{hauenstein2012algorithm}}]
Let $P$ be a square system of polynomials and suppose that $P'(x)$ is nonsingular at $x\in\CC^n$.  Define
$$
\mu(P,x):=\max\left\{1,\|P\|\|P'(x)^{-1}\Delta_P(x)\|\right\}
$$
where the norm in $\|P'(x)^{-1}\Delta_P(x)\|$ is the operator norm. Then, 
$$
\gamma(P,x)\leq\frac{\mu(P,x)d^{\frac{3}{2}}}{2\|(1,x)\|}.
$$
\end{prop}


\subsection{Bounds on $\gamma$ for general systems}

In this section, we apply the results from  \S\ref{sec:alpha:polynomial} to systems of the form of Equation (\ref{eq:SystemWithIngredients}).  In particular, we call $P$ the part of $F$ consisting of polynomial equations.  We begin by observing that the results in \cite[Theorem 2.3]{hauenstein2017certifying} can be directly generalized to the setting of analytic functions.  In particular, we let
\begin{equation*}
\Delta_F =\left[\begin{array}{cc}
\Delta_P(x)\|P\| & \\
& I_m
\end{array}\right]
\end{equation*} 
be an $(n+m)\times (n+m)$ diagonal matrix.  When $F'$ is invertible at $x\in\CC^{n+m}$, we define 
\begin{equation*}
\mu(F,x):=\max\left\{1,\left\|F'(x)^{-1} 
\Delta_F
\right\|\right\}.
\end{equation*}
By the proof of \cite[Theorem 2.3]{hauenstein2017certifying}, we conclude that
$$
\gamma(F,x)\leq\mu(F,x)\sup_{k\geq 2}\left(\left(\frac{d^{\frac{3}{2}}}{2\|(1,x)\|}\right)^{2(k-1)}\hspace{-.05in}+\sum_{i=1}^m\left|\frac{g_i^{(k)}(x_i)}{k!}\right|^2\right)^{\frac{1}{2(k-1)}}.
$$
By the concavity of the of the $2(k-1)^{\text{th}}$ root, it follows that 
\begin{equation}\label{gamma:bound:analytic}
\gamma(F,x)\leq\mu(F,x)\left(\frac{d^{\frac{3}{2}}}{2\|(1,x)\|}+\sup_{k\geq 2}\sum_{i=1}^m\left|\frac{g_i^{(k)}(x_i)}{k!}\right|^{\frac{1}{k-1}}\right).
\end{equation}
Therefore, we observe that, in order to get a bound on $\gamma$, it is enough to bound $\left|\frac{g_i^{(k)}(t)}{k!}\right|^{\frac{1}{k-1}}$ independently of $k$ for each ingredient $g_i$.  In \cite{hauenstein2017certifying}, Hauenstein and Levandovskyy find a bound on these quantities using a recurrence relation from the defining linear differential equation with constant coefficients.  In this paper, we achieve such a bound via the Cauchy integral theorem.

\begin{lem}\label{lem:oracles}
Suppose that the following two oracles exist:
\begin{enumerate}
\item Given a univariate analytic function $g$ and a point $x\in\CC$ in the domain of $g$, there is an oracle which returns a positive value $R>0$ so that the radius of convergence of a power series for $g$ centered at $x$ is at least $R$.
\item Given a univariate analytic function $g$, a point $x\in\CC$ in the domain of $g$, and a radius $r$, there is an oracle which returns $M$ which is an upper bound on the value of $|g|$ on the closed disk $\overline{D}(x,r)$.
\end{enumerate}
Then, for $k\geq 2$,
$$
\left|\frac{g^{(k)}(t)}{k!}\right|^{\frac{1}{k-1}}\leq \frac{1}{r}\max\left\{1,\frac{M}{r}\right\}.
$$
\end{lem}
\begin{proof}
Using Cauchy's integral theorem, we have that 
$$
\frac{|g^{(k)}(x)|}{k!}=\left|\int_0^1\frac{g(x+re^{2\pi it})}{(re^{2\pi it})^k}dt\right|\leq \frac{M}{r^k}.
$$
Therefore,
$$
\left|\frac{g^{(k)}(x)}{k!}\right|^{\frac{1}{k-1}}\leq\frac{1}{r}\left(\frac{M}{r}\right)^{\frac{1}{k-1}}.
$$
Since $k\geq 2$, the $(k-1)^{\text{th}}$ root of $\frac{M}{r}$ is bounded as in the statement of the lemma.
\end{proof}

From this bound, which is independent of $k$, we can now derive a bound on $\gamma(F,x)$.  By substituting this formula into Inequality (\ref{gamma:bound:analytic}), we have a bound on $\gamma(F,x)$.  We collect this result in the following theorem:
\begin{thm}\label{theorem:gammabound}
Let $U\subset\CC^{n+m}$ and consider a system $F:U\rightarrow\CC^{n+m}$ as in Equation (\ref{eq:SystemWithIngredients}) and let $x\in\CC^{n+m}$.  Moreover, assume that there exist oracles as in the statement of Lemma \ref{lem:oracles}.  For each $g_i$, let $R_i$ be a positive lower bound on the radius of convergence for $g_i$ at $x_i$ (given by the first oracle in Lemma \ref{lem:oracles}).  For each $i$, fix $0<r_i<R_i$ to be a positive value strictly less than the radius of convergence.  Then, using the second oracle in Lemma \ref{lem:oracles}, let $M_i$ be an upper bound on $|g_i|$ on the closed disk $\overline{D}(x_i,r_i)$.  For each $i$, let 
$$
C_i=\frac{1}{r_i}\max\left\{1,\frac{M_i}{r_i}\right\}.
$$
Then,
$$
\gamma(F,x)\leq\mu(F,x)\left(\frac{d^{\frac{3}{2}}}{2\|(1,x)\|}+\sum_{i=1}^mC_i\right).
$$
\end{thm}

\begin{rmk}\label{rem:alpharadius}
We remark that the choice of $r_i$ is critically important in this computation.  When $r_i$ is small, $\frac{1}{r_i}$ becomes large, and when $r_i$ is quite large, the disk $\overline{D}(x,r_i)$ approaches a singularity of $g_i$, so $M_i$ is quite large.  Therefore, different choices of $r_i$ can affect the value of $C_i$ drastically.  We provide experimental data illustrating this issue in \S\ref{sec:examples}.
\end{rmk}

We observe that we may apply the same approach as in Theorem \ref{theorem:gammabound} to both $g'$ and $g''$ to achieve potentially tighter bounds on $\gamma(F,x)$.  We make this explicit in the following corollary:
\begin{cor}\label{cor:boundforIngredients}
Suppose that the conditions of Theorem \ref{theorem:gammabound} hold and, in addition, the oracles in the statement of Lemma \ref{lem:oracles} exist for both $g'$ and $g''$.  Let $M_i'$ and $M_i''$ be upper bounds on $|g_i'|$ and $|g_i''|$, respectively, given by the oracle from Lemma \ref{lem:oracles} on $\overline{D}(x_i,r_i)$.  Then, the $C_i$ in Theorem \ref{theorem:gammabound} can be replaced by
\begin{equation}\label{eq:Cbounds}
C_i=\frac{1}{r_i}\max\left\{1,\min\left\{\frac{M_i}{r_i},\frac{M_i'}{2},\frac{M_i''r_i}{2}\right\}\right\}.
\end{equation}
\end{cor}
\begin{proof}
We illustrate the key step in the computation for $M_i'$; the other cases are similar or appear in Theorem \ref{theorem:gammabound}.  We observe that since $k\geq 2$,
$$
\left|\frac{g_i^{(k)}(x)}{k!}\right|^{\frac{1}{k-1}}\leq
\left|\frac{(g_i')^{(k-1)}(x)}{2(k-1)!}\right|^{\frac{1}{k-1}}\leq\frac{1}{r_i}\left(\frac{M_i}{2r_i}\right)^{\frac{1}{k-1}},
$$
where the second inequality follows from applying the inequality of Lemma \ref{lem:oracles} to $g'$.  By considering the possible magnitudes of $\frac{M_i}{2r_i}$, the desired result follows.
\end{proof}

Based on the discussion above we outline an algorithm to certify a root of the system $F$.

\begin{algorithm} {\textbf{AlphaTest$(F,x,r_i,M_i,M'_i,M''_i)$:}}
\renewcommand{\algorithmicrequire}{\textbf{Input:}}
\renewcommand{\algorithmicensure}{\textbf{Output:}}
\begin{algorithmic}
  \REQUIRE{A differentiable system of functions $F:U\rightarrow \CC^{n+m}$ for an open set $U\subset \CC^{n+m}$, a point $x\in \mathbb{C}^{n+m}$, a positive value $r_i$ such that $0<r_i<R_i$ for each $i$, and upper bounds $M_i,M'_i,M''_i$ on $|g_i|,|g'_i|,|g''_i|$ on the closed disk $\overline{D}(x_i,r_i)$ for each $i$.}
  \ENSURE{The boolean value of a condition that implies ``$x$ is an approximate solution of $F$''.}
  
  \smallskip
  \hrule
  \smallskip

  \STATE Compute constants $\beta(F,x), \mu(F,x)$ and $C_i$ in Corollary \ref{cor:boundforIngredients}.
  \RETURN {$\beta(F,x)\mu(F,x)\left(\frac{d^{\frac{3}{2}}}{2\|(1,x)\|}+\sum\limits_{i=1}^mC_i\right)<\frac{13-3\sqrt{17}}{4}$,}
\end{algorithmic}
\end{algorithm}

In the next section, we show that the oracles required by Lemma \ref{lem:oracles} exist for \Dfinite\ functions.

\begin{rmk}
We observe that the results in this section apply when $\CC$ is replaced by $\RR$.  In particular, real roots are certified using the standard techniques of $\alpha$-theory for real roots.  The derived bounds on $\gamma$, however, use the complex values of the radius of convergence and maximum of the function, not merely the real part.
\end{rmk}


\section{The case of $D$-finite functions}\label{sec:Dfinite}



In this section, we show that the oracles needed in \S\S \ref{sec:Krawczyk} and \ref{sec:alpha-theory} exist for \Dfinite\ functions.  These oracles fall into two classes: evaluating a \Dfinite\ function or finding the radius of convergence of a \Dfinite\ function. We point out that the oracles can be obtained from known software implementations. 

\subsection{Evaluating \texorpdfstring{$\bm{D}$}{D}-finite functions}
The analytic continuation algorithm of Chudnovsky and Chudnovsky, first presented in \cite{david1990computer} and further developed in \cite{van1999fast}, provides an algorithm to approximate the value of a \Dfinite\ function.  In particular, the \texttt{SageMath} \cite{sagemath} package \texttt{ore\_algebra.analytic}~\cite{mezzarobba2016rigorous} uses this technique and provides functions which compute an interval containing the image of a \Dfinite\ function over a point or interval.

The output of this algorithm can be used to calculate intervals or boxes containing $F$ and $F'$ when evaluated at points or over intervals (we note that the derivative of a \Dfinite\ function is also \Dfinite).

\begin{rmk}
In the real case, an alternate approximation method using Chebyshev polynomials is presented in \cite{benoit2017rigorous}.  These methods return Chebyshev polynomials such that, on an interval $I$, the point-wise difference between the polynomial and the prescribed \Dfinite\ function is within a specified error.  By applying interval arithmetic on this polynomial, a \Dfinite\ function can be evaluated on an interval.  An implementation of this approximation is available in \texttt{Maple} \cite{maple1994waterloo} and experimental source code is referenced in \cite{benoit2017rigorous}.
\end{rmk}

\subsection{The radius of convergence for \texorpdfstring{$\bm{D}$}{D}-finite functions}

Mezzarobba and Salvy present an algorithm to compute the majorant series for \Dfinite\ function in \cite{mezzarobba2010effective}.  In this case, the radius of convergence for the majorant series is a lower bound on the radius of convergence for the corresponding \Dfinite\ function.  The majorant series provided in \cite{mezzarobba2010effective} has a particularly simple presentation, where the radius of convergence can be identified by the vanishing of a linear term of in denominator, see \cite[Equation (18)]{mezzarobba2010effective}.  The \texttt{Maple} package \texttt{numGfun} \cite{mezzarobba2010numgfun} and the \texttt{SageMath} \cite{sagemath} package \texttt{ore\_algebra.analytic} provide algorithms for computing this majorant series.  For extensions and details of the majorant series approach, see \cite{mezzino1998leibniz} and \cite{van2003majorants}.

\section{Implementation and experiments}\label{sec:examples}
In this section, as a proof of concept, we provide some computational and experimental results for our certification methods for \Dfinite\ functions, as described in \S\ref{sec:Dfinite}.  Our implementations are in \texttt{SageMath} \cite{sagemath}.
We use the \texttt{ore\_algebra.analytic} package from \cite{mezzarobba2016rigorous} for
evaluation of \Dfinite\ functions (function \texttt{numerical\_solution}) and for estimating the radius of convergence for the majorant series (function \texttt{leading\_coefficient}).  The code and all examples in this section are available at\\
\centerline{
  \url{https://github.com/klee669/DfiniteComputationResults} 
  }

\subsection{Comparison between \texorpdfstring{$\bm{\alpha}$}{alpha}-theory and the Krawczyk method.}

The {\em error function} $\erf(t)$ is a basic example of a \Dfinite\ function which satisfies the following differential equation and initial conditions:
\begin{equation*}
\erf''(t) + 2t \erf'(t) =0, \quad \erf(0)=0,\quad \erf'(0)=\frac{2}{\sqrt{\pi}}.
\end{equation*}
We note that the error function has no singularities in $\mathbb{C}$.  We consider the following square system of equations along with the corresponding square function $F$.
\begin{equation}\label{systemWithErrorFunction}
\left\{\begin{aligned}
t_1^2 +t_2^2 &= 4\\
2\erf(t_1)\erf(t_2) &= 1
\end{aligned}\right\}\;\text{with}\;
F(t_1,t_2,t_3,t_4)= \begin{bmatrix}
t_1^2 + t_2^2 - 4\\
t_3t_4 - \frac{1}{2}\\
t_3 - \erf(t_1)\\
t_4 - \erf(t_2)
\end{bmatrix}.
\end{equation}
Using $\mathtt{Mathematica}$ \cite{Mathematica}, we find the following potential solution to this system of equations:
\begin{equation}\label{eq:potentialSoln}
t=(t_1,t_2,t_3,t_4) = (0.480322,1.94147,0.503058,0.993961).
\end{equation}

Using both $\alpha$-theory and the Krawczyk method, we certify that this point approximates a solution to the system of equations in Equation (\ref{systemWithErrorFunction}). In order to study the accuracy required for the $\alpha$-theory-based and Krawczyk method-based tests, we round the coordinates of the point in Equation (\ref{eq:potentialSoln}) to $d$ decimal places and vary $d$ in our experiments appearing in Table \ref{comparison-alpha}. For Krawczyk method-based tests, we also to specify a region by choosing the box whose side length is $2\times 10^{-d}$ centered at the rounded approximation. Moreover, we choose $F'(m(I))^{-1}$ as the invertible matrix $Y$ in Equation (\ref{eq:krawczykinterval}).
\begin{table}[hbt]
	\begin{tabular}{c||c|c}
		decimal places	& Krawczyk method & $\alpha$-theory\\ 
		\hline
		$0$ & fail & fail \\
		$1$	& pass &  fail	 \\ 
		$2$	& pass & fail \\ 
		$3$	& pass &  pass
	\end{tabular}
	\medskip
	\caption{Comparison between the precision required for the Krawczyk-based and \bm{$\alpha$}-theory-based methods.
        }
	\label{comparison-alpha}
\end{table}
      
For the Krawczyk method, a pass indicates that the generalization of the Newton operator is contractive within the given region using the test described in \S\ref{sec:Krawczyk}.  On the other hand, for the $\alpha$-theory-based test from \S\ref{sec:alpha-theory}, a pass indicates that the approximation is certified to be an approximate solution.  Throughout this example, we use $r=0.4$ for the $\alpha$-theory-based test as that gives (nearly) the best value for $r_i$, cf. Remark \ref{rem:alpharadius}.

We observe that Equation (\ref{systemWithErrorFunction}) is an example of a system which could not be effectively studied using the previous $\alpha$-theory techniques.  We also note that the Krawczyk method succeeds with less precision than the $\alpha$-theory-based test.  This behavior is not surprising as the Krawczyk method has a weaker convergence result and uses less pessimistic estimates in its computation.

\subsection{The radius for the \texorpdfstring{$\bm{\alpha}$}{alpha}-theory-based test.}

In this section, we provide some experimental data illustrating the care that must be taken in choosing the radius from \S\ref{sec:alpha-theory}, see Remark \ref{rem:alpharadius}.  We consider a Bessel function (of order $\nu$) $y(t)=C_1Y_\nu(t) + C_2 J_\nu(t)$.  This function is a \Dfinite\ function satisfying the following differential equation:
\begin{equation*}
t^2 y''(t) + t y'(t) + (t^2 - \nu^2) y(t) = 0.
\end{equation*}
We consider the case where $\nu=9$.  In this case, the Bessel function has a regular singularity at $t=0$, its derivative has singularities at $t=0,\pm 9$, and the second derivative has singularities at $t\approx 0,\pm 8.2923, \pm 9, \pm 9.7076$.  Consider the following system of equations and corresponding system $F$ involving a Bessel function and an error function.
$$
\left\{\begin{aligned}
t_1^2 +t_2^2 &= 61\\
2\erf\left(\frac{1}{2}\left(Y_9(t_2) + J_9(t_2)\right) + t_1\right)\left(Y_9(t_2) + J_9(t_2)\right) &= 11
\end{aligned}\right\}\;\text{with}
$$
$$
F(t_1,t_2,t_3,t_4,t_5)=\begin{bmatrix}
t_1^2 + t_2^2  -  61\\
2 t_4 t_5- 11\\
t_3 - \frac{1}{2} t_5 - t_1\\
t_4 - \erf(t_3)\\
t_5 - \left(Y_9(t_2) + J_9(t_2)\right)
\end{bmatrix}.
$$
Using $\mathtt{Mathematica}$ \cite{Mathematica}, we find the following potential solution to this system of equations:
\ifthenelse{\mainfilecheck{1} > 0}{\begin{multline*}
	t=(t_1,t_2,t_3,t_4,t_5) =\\ (6.27899,4.64481,-0.38382,-0.41274,-13.32563).
	\end{multline*}
}{\begin{equation*}
t=(t_1,t_2,t_3,t_4,t_5) = (6.27899,4.64481,-0.38382,-0.41274,-13.32563).
\end{equation*}
}

We apply the $\alpha$-theory-based method of \S\ref{sec:alpha-theory} in attempt to certify this solution while varying radii using the experimentally found lower bound for the radius of convergence, $R=8.2923$. We summarize our results in Table \ref{alphavalues}.

\vspace{-2mm}

\begin{table}[hbt]
	\centering
		\begin{tabular}{c||c|c|c}		
			radius & $\gamma(F,t)$ & $\alpha(F,t)$ & passes $\alpha$-test?  \\ \hline
			$10^{-6}  R$ & $6.9909\cdot 10^7$ & $1.3078$ & no\\
			$10^{-5}  R$ & $6.9909\cdot 10^6$ & $0.1308$ & yes \\
			$10^{-4}  R$ & $6.9912\cdot 10^5$ & $0.0131$ & yes\\
			$10^{-3} R$ & $2.2485\cdot 10^5$ & $0.0042$ & yes\\ 
			$10^{-2}  R$ & $1.9722\cdot 10^6$ & $0.0369$ & yes\\
			$3\cdot 10^{-2}  R$  & $2.4525\cdot 10^7$ & $0.4588$ & no		
								\end{tabular}
 \medskip
 \caption{$\bm{\gamma(F,t)}$ and $\bm{\alpha(F,t)}$ values depending on radii.
 }
   \label{alphavalues}
\end{table}
\vspace{-6mm}
We observe that, as expected from Remark \ref{rem:alpharadius}, a radius which is either too small or too large (when compared to the distance to the singularity) can result in a need for increased precision in the $\alpha$-theory-based test.

\subsection{Comparing \texorpdfstring{$\bm{\alpha}$}{alpha}-theory-based tests on polynomial-exponential systems} 
In this section, we compare the bounds on $\gamma$ that we derive to those from the polynomial-exponential systems in \cite{hauenstein2012algorithm}.  In particular, we consider the following example in the class of polynomial-exponential systems (which are a special case of polynomial-\Dfinite\ systems):
\begin{equation*}
\left\{e^{4t}=0.0183\right\}\;\text{with}\; 
F(t_1,t_2) = \begin{bmatrix}
t_2 - 0.0183\\[.1cm]
t_2 - e^{4t_1}
\end{bmatrix}.
\end{equation*}
For the approximate solution $(t_1,t_2)=(-1,0.018316)$, we compare the bounds on $\gamma(F,t)$ for the method presented in this paper to the $\gamma$ from \cite{hauenstein2012algorithm}, as computed by \texttt{alphaCertified}.  We separate out the three bounds on $\gamma$ from Corollary \ref{cor:boundforIngredients}.  Figure \ref{fig:gammabound} compares the results from \texttt{alphaCertified} and our method.  There, we see that both theoretically and in our implementation, the computed $\gamma$-value may be less than that in \cite{hauenstein2012algorithm}, as computed by \texttt{alphaCertified}.
We note that, in Figure \ref{fig:gammabound}, the implementation bounds differ from the theoretical bounds because the \texttt{ore\_algebra.analytic} package returns inexact outputs when it evaluates functions over an interval.

\ifthenelse{\mainfilecheck{1} > 0}{\begin{figure}[hbt]
		\centering
		\begin{tikzpicture}[xscale=2.2, yscale=0.022]
		\tikzstyle{every node}=[font=\Small]
		\draw [->] (-.1,0) -- (2,0) coordinate (x axis) node[right] {$r$};
		\draw [->] (0,-.5) -- (0,200) coordinate (y axis) node[above] {$\gamma(F,t)$};
		\foreach \x/\xtext in {0.5, 1, 1.5}
		\draw (\x cm,100pt) -- (\x cm,-100pt) node[anchor=north,fill=white] {$\xtext$};
		\foreach \x in {0.1,0.2,0.3,0.4,0.6,0.7,0.8,0.9,1.1,1.2,1.3,1.4,1.6,1.7,1.8,1.9}
		\draw[very thin, gray] (\x cm,50pt) -- (\x cm,0pt);	
		\foreach \y/\ytext in { 40, 80, 120,160}
		\draw (1pt,\y cm) -- (-1pt,\y cm) node[anchor=east,fill=white] {$\ytext$};
		\foreach \y in {8,16,24,32,48,56,64,72,88,96,104,112,128,136,144,152,168,176,184,192}
		\draw[very thin, gray] (-.5pt,\y cm) -- (0pt,\y cm);	
		1	\draw[very thick, red] (0, 84.1574) -- (2,84.1574) node[below,black] {\color{red} $\gamma_{\texttt{alphaCertified}}$};
		\draw[very thick, dotted] plot[domain=0.21:1.88,samples=300] (\x, { 19.3171 *((1/2.82866)+(max(1,(exp(-4 + 4*\x))/(\x))/(\x))}) node[right,black] {$\gamma_0$};	
		\draw[very thick] plot[domain =0.11:1.05,samples=300] (\x, { 19.3171 *((1/2.82866)+ max(1,(8*\x) * exp(-4 + 4*\x))/(\x))}) node[right,black] {$\gamma_{2}$};	
		\draw[very thick, dashed] plot[domain =0.1:1.5,samples=300] (\x, { 19.3171 *((1/2.82866)+ max(1,2*exp(-4 + 4*\x))/(\x))}) node[right,black] {$\gamma_{1}$};		
		\draw[thick, draw=blue] (0.100000000000000,200)--(0.120000000000000,167.878949410761);
		\draw[thick, draw=blue] (0.120000000000000,167.878949410761)--(0.150000000000000,135.669507704824);
		\draw[thick, draw=blue] (0.150000000000000,135.669507704824)--(0.180000000000000,114.196546567533);
		\draw[thick, draw=blue] (0.180000000000000,114.196546567533)--(0.210000000000000,98.8587171837533);
		\draw[thick, draw=blue] (0.210000000000000,98.8587171837533)--(0.240000000000000,95.0126056937628);
		\draw[thick, draw=blue] (0.240000000000000,95.0126056937628)--(0.270000000000000,113.831915150351);
		\draw[thick, draw=blue] (0.270000000000000,113.831915150351)--(0.300000000000000,132.139202379838);
		\draw[thick, draw=blue] (0.300000000000000,132.139202379838)--(0.330000000000000,154.272383639238);
		\draw[thick, draw=blue] (0.330000000000000,154.272383639238)--(0.360000000000000,60.5141437243047);
		\draw[thick, draw=blue] (0.360000000000000,60.5141437243047)--(0.390000000000000,56.3847281209795);
		\draw[thick, draw=blue] (0.390000000000000,56.3847281209795)--(0.420000000000000,52.8452290324150);
		\draw[thick, draw=blue] (0.420000000000000,52.8452290324150)--(0.450000000000000,49.7776631556591);
		\draw[thick, draw=blue] (0.450000000000000,49.7776631556591)--(0.480000000000000,55.3815509957931);
		\draw[thick, draw=blue] (0.480000000000000,55.3815509957931)--(0.510000000000000,63.7763379333043);
		\draw[thick, draw=blue] (0.510000000000000,63.7763379333043)--(0.540000000000000,71.3991972219532);
		\draw[thick, draw=blue] (0.540000000000000,71.3991972219532)--(0.570000000000000,80.4979834880897);
		\draw[thick, draw=blue] (0.570000000000000,80.4979834880897)--(0.600000000000000,91.3475848547440);
		\draw[thick, draw=blue] (0.600000000000000,91.3475848547440)--(0.630000000000000,104.282502570470);
		\draw[thick, draw=blue] (0.630000000000000,104.282502570470)--(0.660000000000000,119.708918401316);
		\draw[thick, draw=blue] (0.660000000000000,119.708918401316)--(0.690000000000000,138.119651259661);
		\draw[thick, draw=blue] (0.690000000000000,138.119651259661)--(0.720000000000000,160.112480322831);
		\draw[thick, draw=blue] (0.720000000000000,160.112480322831)--(0.750000000000000,186.412633122763);	
		\draw[dotted,thick, blue] (0.75,186.412633)--(0.76,200) node[ black] {\color{blue}$\gamma_{\text{implementation}}\quad\quad$};		
		\end{tikzpicture}
		
                \caption{Comparison of computed $\bm{\gamma}$ values in this paper to those from \texttt{alphaCertified}. $\bm{\gamma_0,\gamma_1,\gamma_2}$ indicate bounds computed by $\bm{\frac{M_i}{r_i}, \frac{M_i'}{2},\frac{M_i''r_i}{2}}$ in (\ref{eq:Cbounds}) respectively. $\bm{\gamma_{\text{implementation}}}$ indicates bounds computed by our implementation.
                } \label{fig:gammabound}	
	\end{figure}
}{\begin{figure}[hbt]
	\centering
	\begin{tikzpicture}[xscale=3, yscale=0.03]
	\tikzstyle{every node}=[font=\Small]
	\draw [->] (-.1,0) -- (2,0) coordinate (x axis) node[right] {$r$};
	\draw [->] (0,-.5) -- (0,200) coordinate (y axis) node[left] {upper bound on $\gamma(F,t)$};
	\foreach \x/\xtext in {0.5, 1, 1.5}
	\draw (\x cm,100pt) -- (\x cm,-100pt) node[anchor=north,fill=white] {$\xtext$};
	\foreach \x in {0.1,0.2,0.3,0.4,0.6,0.7,0.8,0.9,1.1,1.2,1.3,1.4,1.6,1.7,1.8,1.9}
	\draw[very thin, gray] (\x cm,50pt) -- (\x cm,0pt);	
	\foreach \y/\ytext in { 40, 80, 120,160}
	\draw (1pt,\y cm) -- (-1pt,\y cm) node[anchor=east,fill=white] {$\ytext$};
	\foreach \y in {8,16,24,32,48,56,64,72,88,96,104,112,128,136,144,152,168,176,184,192}
	\draw[very thin, gray] (-.5pt,\y cm) -- (0pt,\y cm);	
	1	\draw[very thick, red] (0, 84.1574) -- (2,84.1574) node[right,black] {$\gamma_{\texttt{alphaCertified}}$};
	\draw[very thick, dotted] plot[domain=0.21:1.88,samples=300] (\x, { 19.3171 *((1/2.82866)+(max(1,(exp(-4 + 4*\x))/(\x))/(\x))}) node[right,black] {$\gamma_0$};	
	\draw[very thick] plot[domain =0.11:1.05,samples=300] (\x, { 19.3171 *((1/2.82866)+ max(1,(8*\x) * exp(-4 + 4*\x))/(\x))}) node[right,black] {$\gamma_{2}$};	
	\draw[very thick, dashed] plot[domain =0.1:1.5,samples=300] (\x, { 19.3171 *((1/2.82866)+ max(1,2*exp(-4 + 4*\x))/(\x))}) node[right,black] {$\gamma_{1}$};		
	\draw[thick, draw=blue] (0.100000000000000,200)--(0.120000000000000,167.878949410761);
	\draw[thick, draw=blue] (0.120000000000000,167.878949410761)--(0.150000000000000,135.669507704824);
	\draw[thick, draw=blue] (0.150000000000000,135.669507704824)--(0.180000000000000,114.196546567533);
	\draw[thick, draw=blue] (0.180000000000000,114.196546567533)--(0.210000000000000,98.8587171837533);
	\draw[thick, draw=blue] (0.210000000000000,98.8587171837533)--(0.240000000000000,95.0126056937628);
	\draw[thick, draw=blue] (0.240000000000000,95.0126056937628)--(0.270000000000000,113.831915150351);
	\draw[thick, draw=blue] (0.270000000000000,113.831915150351)--(0.300000000000000,132.139202379838);
	\draw[thick, draw=blue] (0.300000000000000,132.139202379838)--(0.330000000000000,154.272383639238);
	\draw[thick, draw=blue] (0.330000000000000,154.272383639238)--(0.360000000000000,60.5141437243047);
	\draw[thick, draw=blue] (0.360000000000000,60.5141437243047)--(0.390000000000000,56.3847281209795);
	\draw[thick, draw=blue] (0.390000000000000,56.3847281209795)--(0.420000000000000,52.8452290324150);
	\draw[thick, draw=blue] (0.420000000000000,52.8452290324150)--(0.450000000000000,49.7776631556591);
	\draw[thick, draw=blue] (0.450000000000000,49.7776631556591)--(0.480000000000000,55.3815509957931);
	\draw[thick, draw=blue] (0.480000000000000,55.3815509957931)--(0.510000000000000,63.7763379333043);
	\draw[thick, draw=blue] (0.510000000000000,63.7763379333043)--(0.540000000000000,71.3991972219532);
	\draw[thick, draw=blue] (0.540000000000000,71.3991972219532)--(0.570000000000000,80.4979834880897);
	\draw[thick, draw=blue] (0.570000000000000,80.4979834880897)--(0.600000000000000,91.3475848547440);
	\draw[thick, draw=blue] (0.600000000000000,91.3475848547440)--(0.630000000000000,104.282502570470);
	\draw[thick, draw=blue] (0.630000000000000,104.282502570470)--(0.660000000000000,119.708918401316);
	\draw[thick, draw=blue] (0.660000000000000,119.708918401316)--(0.690000000000000,138.119651259661);
	\draw[thick, draw=blue] (0.690000000000000,138.119651259661)--(0.720000000000000,160.112480322831);
	\draw[thick, draw=blue] (0.720000000000000,160.112480322831)--(0.750000000000000,186.412633122763);	
	\draw[dotted,thick, blue] (0.75,186.412633)--(0.76,200) node[ black] {$\gamma_{\text{implementation}}\quad\quad$};		
	\end{tikzpicture}
	
	\caption{Comparison of computed $\bm{\gamma}$ values in this paper to those from \texttt{alphaCertified}. $\bm{\gamma_0,\gamma_1,\gamma_2}$ indicate bounds computed by $\bm{\frac{M_i}{r_i}, \frac{M_i'}{2},\frac{M_i''r_i}{2}}$ in (\ref{eq:Cbounds}) respectively. $\bm{\gamma_{\text{implementation}}}$ indicates bounds computed by the implementation. For some choices of $r$, $\gamma_0,\gamma_1,\gamma_2$ and $\gamma_{\text{implementation}}$ have lower values of bounds than \texttt{alphaCertified}.} \label{fig:gammabound}
\end{figure}
}

\subsection{Application to an optimization problem} 
We also use our implementation to solve an optimization problem involving the perimeters of ellipses.  Suppose that $E_1, E_2$ are ellipses with major axes of lengths $1$ and $2$, respectively, whose perimeters sum to $17$.  Suppose that we want to maximize 
\[e_1A_1 + e_2A_2\]
where $e_i$ is the eccentricity and $A_i$ is the area of $E_i$.  Since the area of an ellipse is the product of $\pi$ and the lengths of its axes, if we let $b_i$ be the length of the minor axis of $E_i$, this maximization problem is equivalent to the following problem:
\begin{align*}
\text{Maximize}\quad & e_1b_1 + 2e_2b_2\\
\text{subject to}\quad  & e_1^2+b_1^2=1\\
& 4e_2^2+b_2^2=4\\
& 4E(e_1)+8E(e_2)=17 
\end{align*}
where $E(t)=\int_0^1\frac{\sqrt{1-t^2x^2}}{\sqrt{1-x^2}}dx$ is the complete elliptic integral of the second kind, which satisfies the differential equation
\[(t-t^3)E''(t)+(1-t^2)E'(t)+tE(t)=0.\]
Liouville \cite{liouville1840memoire} showed that $E(t)$ is not algebraic. We can rewrite this maximization problem as a square system of equations by setting up a Lagrange multiplier system.  Since derivatives of \Dfinite\ functions are still \Dfinite\ functions (for the differential equation of the derivative of \Dfinite\ functions, see \cite{van1999fast}), the square system can be certified by the $\alpha$ theory- and the Krawczyk method-based approaches. In our experiments, we use the approximate solution
\[\begin{array}{l}
(b_1,b_2) = (0.8337853,1.5601133), \vspace{1ex}\\ (e_1,e_2) = (0.5520888,0.6257089)\quad \text{and } \vspace{1ex}\\
(\lambda_1,\lambda_2,\lambda_3) = (-0.3310737,-0.4010663,0.0590727)
\end{array}\]
With the choice of the radius $r=0.01$ and an approximate solution with $7$ digits of precision, the $\alpha$-theory-based test certifies the approximate solution. On the other hand, Krawczyk method-based test requires much less precision, in fact, only $2$ digits of precision and a box of side length $2\times 10^{-2}$ are enough to certify this solution.


\section{Conclusion and possible extensions}\label{sec:conclusion}

In this article, we provided a framework to certify isolated nonsingular solutions of square systems of equations involving analytic functions based on explicitly described oracles. We demonstrated that systems with \Dfinite\ functions as ingredients fall into this framework.
We provide a proof-of-concept implementation of the resulting algorithms based on  Krawczyk method and $\alpha$-theory that uses existing software for \Dfinite\ functions. 

Which method performs better in practice depends on one's application. We note, however, that as the ingredients in one's system become more involved the machinery of $\alpha$-theory becomes harder to make effective and use than Krawczyk method. 

\smallskip

This work addresses the case of nonsingular solutions. The question of certification of singular isolated solutions has been studied in many formulations---see, e.g., \cite{mantzaflaris2011deflation, li2014verified, akoglu2018certifying}---for polynomial systems. We remark that certification of singular solutions of analytic systems would require strong additional assumptions---see, e.g., \cite{lee2019isolation}---and new techniques.  The related problem of certifying clusters of roots has been studied using subdivision---see, e.g., \cite{Becker:2016}---for univariate complex polynomials.

\smallskip

It would be also interesting to extend our methods to the systems with \emph{holonomic functions in many variables}\footnote{For a definition of a holonomic function in many variables and a concise review of a restriction algorithm using the theory of $D$-modules see \cite[\S 6.4]{hibi2014grobner}.}  as ingredients.
Given a restriction algorithm of a holonomic function to a line, one can approximate the value of this function at a point; however, one would need an ability to approximate the values of the original function \emph{over a region}. This subtlety requires effective tools that appear to be available only in the univariate case.   
We note that it may be possible to envision a more general ``effective complex analysis'' via multivariate majorant series as in \cite{van2003majorants}. 
In fact, if such effective  analysis is available, our framework can be extended to Pfaffian functions introduced by Khovanskii in~\cite{khovanskiui1991fewnomials} or, more generally, Noetherian functions (see e.g.,~\cite{gabrielov2004complexity} for a definition).


\medskip 
\noindent {\bf Acknowledgements.} We would like to thank ICERM for accommodating all three authors during the semester on Nonlinear Algebra. We also would like to thank Fr\'ed\'eric Chyzak and Jon Hauenstein for helpful discussions. Research of KL is supported in part by NSF grant CCF-1708884 and DMS-1719968. Research of AL is supported in part by NSF grant DMS-1719968.  Research of MB is supported in part by NSF grant CCF-1527193.

\bibliography{ref}

\begin{thebibliography}{10}

\bibitem{akoglu2018certifying}
Tulay~Ayyildiz Akoglu, Jonathan~D Hauenstein, and Agnes Szanto.
\newblock Certifying solutions to overdetermined and singular polynomial
  systems over q.
\newblock {\em J. Symbolic Comput}, 84:147--171, 2018.

\bibitem{Becker:2016}
Ruben Becker, Michael Sagraloff, Vikram Sharma, Juan Xu, and Chee Yap.
\newblock Complexity analysis of root clustering for a complex polynomial.
\newblock In {\em Proceedings of the ACM on International Symposium on Symbolic
  and Algebraic Computation}, ISSAC '16, pages 71--78, New York, NY, USA, 2016.
  ACM.

\bibitem{benoit2017rigorous}
A.~Benoit, M.~Jolde{\c{s}}, and M.~Mezzarobba.
\newblock Rigorous uniform approximation of d-finite functions using chebyshev
  expansions.
\newblock {\em Mathematics of Computation}, 86(305):1303--1341, 2017.

\bibitem{blum2012complexity}
L.~Blum, F.~Cucker, M.~Shub, and S.~Smale.
\newblock {\em Complexity and real computation}.
\newblock Springer Science \& Business Media, 2012.

\bibitem{bozoki2015seven}
S.~Boz{\'o}ki, T.~L. Lee, and L.~R{\'o}nyai.
\newblock Seven mutually touching infinite cylinders.
\newblock {\em Comput. Geom.}, 48(2):87--93, 2015.

\bibitem{david1990computer}
D.~V. Chudnovsky and G.~V. Chudnovsky.
\newblock Computer algebra in the service of mathematical physics and number
  theory.
\newblock {\em Computers in mathematics}, 125:109, 1990.

\bibitem{gabrielov2004complexity}
A.~Gabrielov and N.~Vorobjov.
\newblock Complexity of computations with pfaffian and noetherian functions.
\newblock {\em Normal forms, bifurcations and finiteness problems in
  differential equations}, pages 211--250, 2004.

\bibitem{hauenstein2017certifying}
J.~D. Hauenstein and V.~Levandovskyy.
\newblock Certifying solutions to square systems of polynomial-exponential
  equations.
\newblock {\em J. Symbolic Comput.}, 79:575--593, 2017.

\bibitem{hauenstein2011alphacertified}
J.~D. Hauenstein and F.~Sottile.
\newblock alphacertified: Software for certifying numerical solutions to
  polynomial equations.
\newblock {\em Available at {\tt
  math.tamu.edu/\~{}sottile/research/stories/alphaCertified}}, 2011.

\bibitem{hauenstein2012algorithm}
J.~D. Hauenstein and F.~Sottile.
\newblock Algorithm 921: alphacertified: certifying solutions to polynomial
  systems.
\newblock {\em ACM Trans. Math. Software}, 38(4):28, 2012.

\bibitem{hibi2014grobner}
T.~Hibi.
\newblock {\em Gr{\"o}bner bases: Statistics and software systems}.
\newblock Springer Science \& Business Media, 2014.

\bibitem{Mathematica}
Wolfram~Research{,} Inc.
\newblock Mathematica, {V}ersion 11.3, 2018.
\newblock Champaign, IL.

\bibitem{khovanskiui1991fewnomials}
A.~Khovanskii.
\newblock {\em Fewnomials}, volume~88.
\newblock American Mathematical Soc., 1991.

\bibitem{krawczyk1969newton}
R.~Krawczyk.
\newblock Newton-algorithmen zur bestimmung von nullstellen mit
  fehlerschranken.
\newblock {\em Computing}, 4(3):187--201, 1969.

\bibitem{MR783635}
S.~Lang.
\newblock {\em Real analysis}.
\newblock Addison-Wesley Publishing Company, Advanced Book Program, Reading,
  MA, second edition, 1983.

\bibitem{lee2019isolation}
Kisun Lee, Nan Li, and Lihong Zhi.
\newblock On isolation of singular zeros of multivariate analytic systems.
\newblock {\em arXiv preprint arXiv:1904.07937}, 2019.

\bibitem{li2014verified}
Nan Li and Lihong Zhi.
\newblock Verified error bounds for isolated singular solutions of polynomial
  systems.
\newblock {\em SIAM Journal on Numerical Analysis}, 52(4):1623--1640, 2014.

\bibitem{liouville1840memoire}
J.~Liouville.
\newblock {\em M{\'e}moire sur les transcendantes {\'e}lliptiques de
  premi{\`e}re et de seconde esp{\`e}ce, consid{\'e}r{\'e}es comme fonctions de
  leur module}.
\newblock 1840.

\bibitem{mantzaflaris2011deflation}
Angelos Mantzaflaris and Bernard Mourrain.
\newblock Deflation and certified isolation of singular zeros of polynomial
  systems.
\newblock In {\em Proceedings of the 36th international symposium on Symbolic
  and algebraic computation}, pages 249--256. ACM, 2011.

\bibitem{maple1994waterloo}
Maplesoft.
\newblock Maple (2018).
\newblock {\em a division of Waterloo Maple Inc., Waterloo, Ontario}, 2018.

\bibitem{mezzarobba2010numgfun}
M.~Mezzarobba.
\newblock Numgfun: a package for numerical and analytic computation with
  d-finite functions.
\newblock In {\em Proceedings of the 2010 International Symposium on Symbolic
  and Algebraic Computation}, pages 139--145. ACM, 2010.

\bibitem{mezzarobba2016rigorous}
M.~Mezzarobba.
\newblock Rigorous multiple-precision evaluation of d-finite functions in
  sagemath.
\newblock Technical Report 1607.01967, arXiv, 2016.

\bibitem{mezzarobba2010effective}
M.~Mezzarobba and B.~Salvy.
\newblock Effective bounds for p-recursive sequences.
\newblock {\em J. Symbolic Comput.}, 45(10):1075--1096, 2010.

\bibitem{mezzino1998leibniz}
M.~Mezzino and M.~Pinsky.
\newblock Leibniz's formula, cauchy majorants, and linear differential
  equations.
\newblock {\em Math. Mag.}, 71(5):360--368, 1998.

\bibitem{moore2009introduction}
R.~E. Moore, R.~B. Kearfott, and M.~J. Cloud.
\newblock {\em Introduction to interval analysis}, volume 110.
\newblock Siam, 2009.

\bibitem{MooreKioustelidis}
R.~E. Moore and J.~B. Kioustelidis.
\newblock A simple test for accuracy of approximate solutions to nonlinear (or
  linear) systems.
\newblock {\em SIAM Journal on Numerical Analysis}, 17(4):521--529, 1980.

\bibitem{computermethods:range}
H.~Ratschek and J.~Rokne.
\newblock {\em Computer Methods for the Range of Functions}.
\newblock Ellis Horwood Limited, 1984.

\bibitem{shub2000complexity}
M.~Shub and S.~Smale.
\newblock Complexity of bezout's theorem. i: geometric aspects.
\newblock In {\em The Collected Papers of Stephen Smale: Volume 3}, pages
  1359--1401. World Scientific, 2000.

\bibitem{smale1986newton}
S.~Smale.
\newblock Newton's method estimates from data at one point.
\newblock {\em The Merging of Disciplines: New Directions in Pure, Applied, and
  Computational Mathematics}, 1986.

\bibitem{sagemath}
{The Sage Developers}.
\newblock {\em {S}ageMath, the {S}age {M}athematics {S}oftware {S}ystem
  ({V}ersion 8.3)}, 2018.
\newblock {\tt http://www.sagemath.org}.

\bibitem{van1999fast}
J.~van~der Hoeven.
\newblock Fast evaluation of holonomic functions.
\newblock {\em Theoret. Comput. Sci.}, 210(1):199--215, 1999.

\bibitem{van2003majorants}
J.~van~der Hoeven.
\newblock Majorants for formal power series.
\newblock Technical Report 2003-15, Universit\'e Paris-Sud, Orsay, France,
  2003.

\end{thebibliography}
\bibliographystyle{plain}

\end{document}